\documentclass[a4paper,onecolumn,accepted=2026-04-02]{quantumarticle}
\pdfoutput=1

\let\originalleft\left
\let\originalright\right
\renewcommand{\left}{\mathopen{}\mathclose\bgroup\originalleft}
\renewcommand{\right}{\aftergroup\egroup\originalright}

\usepackage{url}
\usepackage{amsmath, amssymb, amsthm}
\usepackage{enumerate,enumitem}
\usepackage{autobreak, bbm, caption, comment, graphicx, mathtools, sidecap, tikz, tikzsymbols, thmtools, thm-restate, verbatim, xcolor, xspace}

\usepackage[numbers,sort&compress]{natbib}

\usepackage[backref=page]{hyperref}
\hypersetup{colorlinks, allcolors=blue, linktocpage, breaklinks}
\usepackage[capitalize]{cleveref}

\allowdisplaybreaks

\theoremstyle{plain}
\newtheorem{thm}{Theorem}[section]
\newtheorem{lem}[thm]{Lemma}
\newtheorem{cor}[thm]{Corollary}

\newtheorem{obs}[thm]{Observation}
\newtheorem{clm}[thm]{Claim}

\crefname{lem}{Lemma}{Lemmas}
\crefname{thm}{Theorem}{Theorems}
\crefname{cor}{Corollary}{Corollaries}
\crefname{prp}{Proposition}{Propositions}

\theoremstyle{definition}
\newtheorem{dfn}[thm]{Definition}
\newtheorem{que}[thm]{Question}

\crefname{que}{Question}{Questions}

\newcommand*{\C}{\ensuremath{\mathbb{C}}}

\newcommand*{\bra}[1]{\ensuremath{\langle #1|}}
\newcommand*{\ket}[1]{\ensuremath{|#1\rangle}}

\newcommand*{\bigket}[1]{\left|#1\right\rangle}
\newcommand*{\ip}[2]{\langle #1 | #2 \rangle}
\newcommand*{\ketbra}[2]{|#1\rangle\!\langle #2|}
\newcommand*{\kb}[1]{\ketbra{#1}{#1}}

\newcommand*{\acz}{$\mathsf{AC^0}$\xspace}

\newcommand*{\qaczf}{$\mathsf{QAC_f^0}$\xspace}
\newcommand*{\qacf}{$\mathsf{QAC_f}$\xspace}
\newcommand*{\qnc}{$\mathsf{QNC}$\xspace}

\newcommand{\Paren}[1]{\left(#1\right)}
\newcommand*{\Mag}[1]{\left| #1 \right|}

\newcommand*{\adj}[1]{#1^\dagger}

\newcommand*{\bits}{\{0,1\}}
\newcommand*{\eps}{\varepsilon}
\newcommand*{\E}{\mathbb{E}}

\newcommand*{\Ind}[1]{\mathbbm1_{#1}}
\newcommand*{\norm}[1]{\|#1\|}
\newcommand*{\Norm}[1]{\left\|#1\right\|}
\newcommand*{\poly}{\mathrm{poly}}
\newcommand*{\zs}{0\dotsc0}
\DeclarePairedDelimiter{\ceil}{\lceil}{\rceil}

\newcommand*\ot[1]{\tilde O \left( #1 \right)}
\newcommand*\oo[1]{O \left( #1 \right)}
\newcommand*\cube[1]{\bits^{#1}}
\newcommand*\es\epsilon
\newcommand*\reg\mathsf

\newcommand*\fip[2]{\left\langle #1, #2 \right\rangle}
\newcommand*\tr[1]{\operatorname{tr} \left( #1 \right)}

\newcommand*{\prz}{\mathrm{Pr}}
\newcommand*{\pr}[1]{\mathrm{Pr}(#1)}

\newcommand*{\PRs}[2]{\mathrm{Pr}_{#1}\left(#2\right)}

\usetikzlibrary{arrows, calc, positioning, shapes}

\newcommand*{\chan}[1]{\mathcal C_{#1}}
\newcommand*{\mrm}[1]{\mathrm{#1}}
\newcommand{\Brac}[1]{\left[#1\right]}
\begin{document}
	
\title{Query and Depth Upper Bounds for Quantum Unitaries via Grover Search}
\author{Gregory Rosenthal}
\email{grosenth@uwaterloo.ca}
\affiliation{Institute for Quantum Computing, University of Waterloo}

\begin{abstract}
	We prove that any $n$-qubit unitary can be implemented (i) approximately in time $\ot{2^{n/2}}$ with query access to an appropriate classical oracle, and also (ii) exactly by a circuit of depth $\ot{2^{n/2}}$ with one- and two-qubit gates and $2^{O(n)}$ ancillae. The proofs involve similar reductions to Grover search. The proof of (ii) also involves a linear-depth construction of arbitrary quantum states using one- and two-qubit gates (in fact, this can be improved to constant depth with the addition of fanout and generalized Toffoli gates) which may be of independent interest. We also prove a matching $\Omega\Paren{2^{n/2}}$ lower bound for (i) and (ii) for a certain class of implementations.
\end{abstract}

\maketitle

\section{Introduction}\label{s:intro}
This paper addresses two seemingly disparate questions in quantum circuit complexity via a common proof technique. The first of these questions is as follows:

\begin{que}[The unitary synthesis problem~\cite{Aar16,AK07}] \label{q:usp}
	Is there a polynomial-time quantum algorithm $A$ such that for every unitary $U$, there exists a classical oracle $f$ such that $A^f$ approximately implements $U$?
\end{que}

By $A^f$ we mean $A$ with query access to the Boolean function $f$. Note that \cref{q:usp} is concerned with the \emph{overall} runtime rather than just the number of queries. \cref{q:usp} was posed by Aaronson and Kuperberg~\cite{AK07} and named the ``unitary synthesis problem" by Aaronson~\cite{Aar16}. An affirmative answer would imply that to obtain a small quantum circuit for a unitary $U$, it suffices to give an efficient algorithm for computing the oracle $f$, which is interesting because more is known about Boolean function complexity than is known about quantum circuit complexity. Aaronson~\cite{Aar21} and Lombardi, Ma and Wright~\cite{LMW23} discuss this motivation in the context of certain physically motivated unitaries and quantum cryptography respectively.

We also consider the following question:

\begin{que}\label{q:depth}
	Given $n$, what is the minimum circuit depth required to exactly implement a worst-case $n$-qubit unitary using one- and two-qubit gates?
\end{que}

The depth of a circuit is the number of layers of gates in it. Circuit depth corresponds to parallel computation time, and quantum circuits of higher depth are believed to be more difficult to physically implement. An upper bound for the unitary synthesis problem implies a similar depth upper bound for \emph{approximately} implementing worst-case unitaries, because any $n$-bit function can be computed by a circuit of depth $O(n)$ (e.g.\ a CNF or DNF), but otherwise \cref{q:usp,q:depth} are not obviously related.

We prove $\ot{2^{n/2}}$ upper bounds for both \cref{q:usp,q:depth}, which improve on the previously best known upper bounds by a constant factor in the exponent. This is discussed in more detail in \cref{s:ub-usp,s:ub-ldqc} respectively. Both upper bounds are proved using a reduction from the task of implementing a unitary $U$ to that of implementing what we call a \emph{$U$-column-constructor}, or $U$-CC for short:

\begin{dfn}[$U$-CC] \label{U-qram-def}
	Given an $n$-qubit unitary $U$, call a unitary $A$ acting on $m \ge 2n$ qubits a \emph{$U$-column-constructor ($U$-CC)} if $A \ket{x, 0^{m-n}} = \ket{x} \otimes U\ket{x} \otimes \ket{0^{m-2n}}$ for all $x \in \cube n$.
\end{dfn}

Informally, controlled on an input string $x \in \cube n$, a $U$-CC constructs the corresponding output state $U \ket x$ of $U$ in a separate register;\footnote{Here and throughout this paper, we allow ancillae that start in the all-zeros state and are required to end in the all-zeros state.} if this separate register is not initialized to the all-zeros state then a $U$-CC's behavior is unspecified subject to unitarity. Using a zero-error variant of Grover search we prove the following, where by $C^A$ we mean $C$ with $A$ and $\adj A$ oracles:

\begin{thm} \label{sc-main}
	There is a uniform family $(C_{n,m})_{n,m}$ for $m \ge 2n$ of quantum circuits, each making $\oo{2^{n/2}}$ queries to an $m$-qubit quantum oracle, such that for all $n$-qubit unitaries $U$ and all $m$-qubit $U$-CCs $A$ it holds that $C_{n,m}^A$ implements $U$.
\end{thm}

To see why \cref{sc-main} is nontrivial, suppose that we wish to apply a unitary $U$ on the input state $\sum_{x \in \cube n} \alpha_x \ket x$. A natural first step is to query a $U$-CC to obtain the state $\sum_{x \in \cube n} \alpha_x \ket{x} \otimes U\ket x$. But now to obtain $U \sum_{x \in \cube n} \alpha_x \ket x$, it is necessary to uncompute $\ket x$ in superposition controlled on $U \ket x$.

The analogous ``state synthesis problem" for constructing quantum states using a classical oracle has polynomial-time solutions~\cite{Aar16,Ros24}, and in this paper we prove that any $n$-qubit state can be constructed in $O(n)$ depth using $\ot{2^n}$ ancillae. These upper bounds for constructing states generalize to implementing $U$-CCs, and plugging these implementations of $U$-CCs into the algorithm from \cref{sc-main} yields $\ot{2^{n/2}}$ upper bounds for \cref{q:usp,q:depth}.

Finally, we prove a matching $\Omega\Paren{2^{n/2}}$ query lower bound for \cref{sc-main} when $U$ is Haar random and the $U$-CC is defined appropriately, implying that new techniques are needed to make further progress on \cref{q:usp,q:depth}. This is discussed further in \cref{s:ilb}.

Below we state our results informally; more precise statements will be given in subsequent sections.
\subsection{Upper bound for the unitary synthesis problem} \label{s:ub-usp}

We prove the following:

\begin{restatable}{thm}{rclass} \label{class}
	There is an $\ot{2^{n/2}}$-time quantum algorithm $A$ such that for every $n$-qubit unitary $U$, there exists a classical oracle $f$ such that $A^f$ approximately implements $U$ to within exponentially small error.
\end{restatable}

\cref{class} is the first nontrivial upper or lower bound for the unitary synthesis problem. For comparison, every $n$-qubit unitary can be implemented using $\ot{2^{2n}}$ one- and two-qubit gates~\cite[Section 4.5]{NC10}, and hence by the Solovay-Kitaev theorem~\cite{DN06} can be approximated to within exponentially small error using $\ot{2^{2n}}$ gates from a finite gate set. Thus the trivial algorithm in which the oracle encodes the description of a circuit for approximating $U$ runs in time $\ot{2^{2n}}$. Irani, Natarajan, Nirkhe, Rao and Yuen~\cite[Section 7.2]{INN+22} also observed that \emph{with postselection} there is a polynomial-time solution to the unitary synthesis problem, using the Choi--Jamiołkowski isomorphism and quantum teleportation. Subsequently to our work, Lombardi et al.~\cite{LMW23} proved that there is no one-query solution to the unitary synthesis problem using $o(2^n)$ qubits.

\subsection{Circuit depth upper bounds for states and unitaries} \label{s:ub-ldqc}

For brevity we assign the following name to the ``standard" quantum circuit model:

\begin{dfn}[\qnc circuits]
	A \qnc circuit is a quantum circuit consisting of one- and two-qubit gates.
\end{dfn}

We prove the following:

\begin{restatable}{thm}{rdepth} \label{depth}
	Every $n$-qubit unitary can be implemented by a \qnc circuit of depth $\ot{2^{n/2}}$ with $\ot{2^{2n}}$ ancillae.
\end{restatable}

Sun, Tian, Yang, Yuan and Zhang~\cite{Sun+21} proved that every $n$-qubit unitary can be implemented by a \qnc circuit of depth $\ot{2^n}$ with $\oo{2^n}$ ancillae, compared to which the circuit from \cref{depth} has lower depth but more ancillae. More generally, Sun et al.~\cite{Sun+21} proved that for $m \le 2^n$, any $n$-qubit unitary can be implemented by a \qnc circuit of size (i.e.\ number of gates) $\oo{4^n}$ and depth $\ot{4^n/m}$ with $m$ ancillae. In followup work, Yuan and Zhang~\cite{YZ23} generalized our proof of \cref{depth} to show that for $2^n \le m \le 4^n$, any $n$-qubit unitary can be implemented by a \qnc circuit of depth $\ot{2^{3n/2} m^{-1/2}}$ with $m$ ancillae; when $m = 4^n$ this matches \cref{depth} up to $\poly(n)$ factors.

The size of a circuit is trivially at most its depth times number of qubits acted on, so the circuit from \cref{depth} has size $\ot{2^{2.5 n}}$. This raises the following question:

\begin{que}
	Can every $n$-qubit unitary be implemented by a \emph{single} \qnc circuit that is \emph{both} of size $\ot{4^n}$ and depth $\ot{2^{n/2}}$?
\end{que}

Our proof of \cref{depth} uses a low-depth construction of quantum states which may be of independent interest, and of which we first state the following corollary:

\begin{cor} \label{state-qnc}
	Every $n$-qubit state can be constructed by a \qnc circuit of depth $O(n)$ with $\ot{2^n}$ ancillae.
\end{cor}

Sun et al.~\cite{Sun+21} and Zhang, Li and Yuan~\cite{ZLY22} independently proved \cref{state-qnc}, respectively shortly before and shortly after we did, and with just $\oo{2^n}$ ancillae and $\oo{2^n}$ size. Yuan and Zhang~\cite{YZ23} proved in followup work that every $n$-qubit state can be constructed by a \qnc circuit of size $\oo{2^n}$ and depth $\oo{n + \frac {2^n} {n+m}}$ using $m \ge 0$ ancillae, and that these size and depth upper bounds are tight for all $n,m$. This improves on a slightly weaker tradeoff of Sun et al.~\cite{Sun+21}, and the proof of Yuan and Zhang's~\cite{YZ23} upper bound cites ideas from our proof of \cref{state-qnc}.

\cref{state-qnc} follows from a \emph{constant}-depth construction of quantum states over a larger gate set. The following class was defined by Green, Homer, Moore and Pollett~\cite{Gre+02}:

\begin{restatable}[\qaczf~\cite{Gre+02}]{dfn}{rdefqacz}
	A \qacf circuit is a quantum circuit consisting of arbitrary one-qubit gates, as well as \emph{generalized Toffoli gates} of arbitrary arity defined by
	\begin{equation*}
		\ket{b,x} \mapsto \bigket{b \oplus \prod_{j=1}^n x_j, x} \quad \text{for} \quad b \in \bits, x = (x_1, \dotsc, x_n) \in \cube n,
	\end{equation*}
	and \emph{fanout gates} of arbitrary arity defined by
	\begin{equation*}
		\ket{b,x} \mapsto \ket{b, x \oplus b^n} \quad \text{for} \quad b \in \bits, x \in \cube n.
	\end{equation*}
	A \qaczf circuit is a constant-depth \qacf circuit.
\end{restatable}

Analogously to in classical circuit complexity, one motivation for studying restricted quantum circuit classes such as \qaczf circuits is that they seem potentially easier to prove lower bounds against than general \qnc circuits. \qaczf circuits can trivially simulate \acz circuits~\cite{Gre+02} (i.e.\ constant-depth Boolean circuits with NOT gates and unbounded-fanin AND and OR gates), and in fact are even more powerful than their classical counterparts, because polynomial-size \qaczf circuits can also compute the majority function~\cite{HS05,TT16} whereas \acz circuits require exponential size to do so~\cite{Has86,Juk12}. We prove the following:

\begin{restatable}{thm}{rstateqacf} \label{state-qacf}
	Every $n$-qubit state can be constructed by a \qaczf circuit with $\ot{2^n}$ ancillae.
\end{restatable}

We will see that \qnc circuits of logarithmic depth can efficiently simulate \qaczf circuits, so \cref{state-qacf} implies \cref{state-qnc}.

It is well known that every function from $n$ bits to one bit can be computed by a DeMorgan circuit of depth $O(n)$ and size $\oo{2^n/n}$, and that for most functions this upper bound is tight~\cite{Juk12,Lup58,Sha49}.\footnote{The $\Omega(n)$ depth lower bound follows from the $\Omega\Paren{2^n/n}$ size lower bound, because any Boolean circuit of depth $d$ has size less than $2^d$.} The above results can be seen as progress toward analogous statements about the quantum circuit complexity of constructing quantum states and implementing unitary transformations.

\subsection{Lower bound for implementing $U$ given a $U$-CC} \label{s:ilb}

We prove the following, where by $C^A$ we mean $C$ with $A$ and $\adj A$ oracles:

\begin{thm} \label{sc-lb}
	For all sequences of quantum circuits $(C_n)_n$ making $o\Paren{2^{n/2}}$ queries to a $2n$-qubit quantum oracle, with probability $1-o(1)$ over a Haar random $n$-qubit unitary $U$, there exists a $2n$-qubit $U$-CC $A$ such that $C_n^A$ is (in some sense) almost maximally far from implementing $U$.
\end{thm}

\cref{sc-lb} matches the upper bound from \cref{sc-main}. Since any $U$-CC tensored with the identity is also a $U$-CC, we may replace ``$2n$-qubit" with ``$m$-qubit" in \cref{sc-lb} for any $m \ge 2n$. However, some restrictions on $A$ are still necessary for a lower bound such as \cref{sc-lb} to hold, at least if we allow $A$ to act on more than $2n$ qubits. For example, the unitary $A$ defined by
\begin{equation*}
	\forall x,y \in \cube n, b \in \bits: \, A \ket{x,y,b} =
	\begin{cases}
		\ket{x} \otimes U \ket{x \oplus y} \otimes \ket0 &\text{if } b=0 \\
		U\ket{x} \otimes \ket{y} \otimes \ket1 &\text{if } b=1
	\end{cases}
\end{equation*}
is a $U$-CC, and can trivially be used to implement $U$ when applied with $b=1$.

It is well known that unstructured search on a list of length $N$ requires $\Omega(\sqrt N)$ quantum queries~\cite{NC10}, but this does not immediately imply that \cref{sc-main} is tight, since there also exist algorithms that do not simulate unstructured search. As we will explain more precisely when we prove \cref{sc-lb}, it also takes $\Omega(\sqrt N)$ quantum queries to compute $\sigma^{-1}(1)$ given query access to a permutation $\sigma$ of $\{1, \dotsc, N\}$~\cite{Amb02,Nay11}, and this almost immediately implies an $\Omega\Paren{2^{n/2}}$ lower bound for \cref{sc-main} when $U$ is a permutation matrix and $A$ is defined appropriately. However this example is unsatisfying if our ultimate goal is to prove lower bounds for \cref{q:usp,q:depth}, since $n$-qubit permutation matrices can be efficiently synthesized with a classical oracle and also implemented by a \qnc circuit of depth $O(n)$.\footnote{On input $x$, first compute $\sigma(x)$ (by querying the oracle, or by simulating an appropriate Boolean circuit, depending on the model of computation), and then uncompute $x$ given $\sigma(x)$ by running a similar procedure in reverse.} In contrast, if \emph{any} family of unitaries is hard to implement in the sense of \cref{q:usp,q:depth}, then Haar random unitaries are also hard to implement for the following reason:

\begin{obs} \label{wcac}
	Any fixed unitary $U$ can be written as $U = UR \cdot \adj R$ pointwise where $R$ is Haar random, so since $UR$ and $\adj R$ are also Haar random, the task of implementing $U$ reduces to that of successively implementing two (dependent) Haar random unitaries.
\end{obs}

This reduction, along with the $U$-CCs from our proofs of \cref{class,depth}, shows that an improved upper bound for \cref{sc-main} in the case where $U$ is Haar random would imply improved upper bounds for \cref{q:usp,q:depth} in the general case. (Provided that in this hypothetical improvement to \cref{sc-main}, the complexity of the non-query operations is not too large.) However, \cref{sc-lb} rules out this approach.

\cref{sc-lb} does, however, leave open the possibility of obtaining a tighter upper bound for \cref{q:depth,q:usp} by reducing to some \emph{other} ``generalized CC":

\begin{que} \label{op:alt-usp}
	Is there a sequence $(C_n)_n$ of quantum circuits, each making $o\Paren{2^{n/2}}$ queries to a $(p(n) + q(n))$-qubit quantum oracle where $p(n), q(n) = \poly(n)$, such that for all $n$-qubit unitaries $U$ there exists a family of $q(n)$-qubit states $\Psi = (\ket{\psi_x})_{x \in \cube{p(n)}}$ such that for all ``$\Psi$-CCs" $A$ (i.e.\ $A\ket{x, \zs} = \ket{x} \ket{\psi_x}$) it holds that $C_n^A$ implements $U$?
\end{que}

If $\Psi = (\ket{f(x)})_{x \in \cube{p(n)}}$ for some function $f: \cube{p(n)} \to \bits$ then a $\Psi$-CC computes $f$, so \cref{op:alt-usp} is essentially a rephrasing of the unitary synthesis problem with a more modest runtime requirement. However, this rephrasing suggests an approach to proving lower bounds for the unitary synthesis problem by considering increasingly general classes of state families $\Psi$.

At a high level, we prove \cref{sc-lb} by using \cref{wcac} to reduce to the previously mentioned lower bound for the case where $U$ is a permutation matrix.

\subsection{Organization and preliminaries}

In \cref{sec:qram} we prove upper and lower bounds for implementing $U$ given query access to a $U$-CC and its inverse. In \cref{s:class} we prove an upper bound for the unitary synthesis problem. Finally in \cref{sec:depth} we prove circuit depth upper bounds for constructing states and implementing unitaries.

We denote the $n$-qubit identity transformation by $I_n$, or $I$ when $n$ is implicit. We write $\norm\cdot$ to denote the 2-norm of a vector or the operator norm of a matrix.

\section{Bounds for implementing $U$ given a $U$-CC} \label{sec:qram}
In \cref{s:main,s:lb} we prove upper and lower bounds respectively on the complexity of implementing a unitary $U$, given query access to a $U$-CC and its inverse. First we define the quantum query model more precisely. By a \emph{quantum circuit making $k$ queries to an $n$-qubit quantum oracle}, we mean a circuit of the form $C = C_k Q_k C_{k-1} Q_{k-1} \dotsb C_0$ where each $C_j$ is a unitary and each $Q_j$ is a placeholder for either a ``forward" or ``backward" query. For an $n$-qubit unitary $A$, by $C^A$ we mean the unitary defined by substituting $A$ and $\adj A$ respectively for the forward and backward queries in $C$. Claims about the quantum circuit complexity of $C$ are in reference to the circuit $C_k C_{k-1} \dotsb C_0$ defined by removing the queries from $C$. Let $\adj C = \adj C_0 \adj Q_1 \adj C_1 \adj Q_2 \dotsb \adj C_k$, where the ``conjugate transpose" of the forward query symbol is the backward query symbol and vice versa, and note that $\Paren{\adj C}^A = \adj{\Paren{C^A}}$.
\subsection{Upper bound} \label{s:main}

We will use a variant of Grover search that finds the marked string with certainty rather than just with high probability:

\begin{restatable}{lem}{regs} \label{egs}
	There is a uniform sequence of \qacf circuits $(G_n)_n$---each of depth $\oo{2^{n/2}}$, making $\oo{2^{n/2}}$ queries, and acting on $n+1$ qubits---such that for all $x \in \cube n$ it holds that $G_n^{I - 2\kb{x,1}} \ket{0^{n+1}} = \ket{x, 0}$.
\end{restatable}

Imre and Bal\'azs~\cite{IB05} survey several proofs of the query upper bound from \cref{egs} in detail, and below we include an alternate proof of \cref{egs} due to Wiebe~\cite{Wie21a}.

\begin{proof}
	Let $x \in \cube n$ denote the marked string that we wish to find, and let
	\begin{align*}
		&t = \left\lceil \frac\pi4 2^{n/2} \right\rceil,
		&\theta = \frac{\pi/2} {2t+1},
		&&p = 2^n \sin^2 \theta.
	\end{align*}
	Since $p \le 2^n \theta^2 \le 2^n \left( \frac{\pi/2} {2t} \right)^2 \le 1$ we can define states
	\begin{align*}
		\ket{\psi_0} = \ket{+^n} \otimes \left( \sqrt{1-p} \ket0 + \sqrt{p} \ket1 \right),
		\qquad
		\ket{\psi_t} = \left( (2\kb{\psi_0} - I) (I - 2\kb{x,1}) \right)^t \ket{\psi_0}
	\end{align*}
	where $\ket+ = \frac{\ket0 + \ket1}{\sqrt 2}$. Since $\ip{x,1}{\psi_0} = 2^{-n/2} \sqrt{p} = \sin\theta$, we may write $\ket{\psi_0} = \cos\theta \ket\alpha + \sin\theta \ket\beta$ where $\ket\beta = \ket{x,1}$ and $\ket\alpha$ is a superposition of standard basis states besides $\ket{x,1}$. In this basis, 
	\begin{align*}
		\ket{\psi_t}
		&= \Paren{\begin{pmatrix} 2 \cos^2\theta - 1 & 2 \cos\theta \sin\theta \\ 2 \cos\theta \sin\theta & 2 \sin^2\theta - 1 \end{pmatrix} \begin{pmatrix} 1 & 0 \\ 0 & -1 \end{pmatrix}}^t \begin{pmatrix} \cos\theta \\ \sin\theta \end{pmatrix}
		= \begin{pmatrix} \cos 2\theta & -\sin 2\theta \\ \sin 2\theta & \cos 2\theta \end{pmatrix}^t \begin{pmatrix} \cos\theta \\ \sin\theta \end{pmatrix} \\
		&= \begin{pmatrix} \cos ((2t+1) \theta) \\ \sin ((2t+1) \theta) \end{pmatrix}
		= \begin{pmatrix} \cos (\pi/2) \\ \sin (\pi/2) \end{pmatrix}
		= \begin{pmatrix} 0 \\ 1 \end{pmatrix},
	\end{align*}
	i.e.\ $\ket{\psi_t} = \ket{x,1}$.
\end{proof}

Now we reduce the task of implementing a unitary $U$ to that of implementing a $U$-CC:

\begin{thm}[formal version of \cref{sc-main}] \label{main}
	There is a uniform family of \qacf circuits $(C_{n,m})_{n,m}$ for $m \ge 2n$---each of depth $\oo{2^{n/2}}$, making $\oo{2^{n/2}}$ queries to an $m$-qubit quantum oracle, and acting on $O(m)$ qubits---such that for all $n$-qubit unitaries $U$ and all $m$-qubit $U$-CCs $A$ it holds that $C_{n,m}^A (I_n \otimes \ket\zs) = U \otimes \ket\zs$.
\end{thm}

\begin{proof}
	By linearity we may assume that the input is a string $x \in \cube n$; our goal is to output $U\ket{x} \otimes \ket\zs$. First apply $A$, yielding $\ket{x} \otimes U\ket{x} \otimes \ket\zs$. The challenge now is to uncompute $\ket x$.
	
	Let
	\begin{equation*}
		C = (A \otimes I_1) \left( I_n \otimes \left(I_{m-n+1} - 2 \kb{0^{m-n}, 1} \right) \right) \left(\adj A \otimes I_1\right),
	\end{equation*}
	and observe that $C$ can be implemented by a \qaczf circuit making two queries. By the definition of $A$ we have that
	\begin{align*}
		C &= (A \otimes I_1) \left(I_{m+1} - 2\sum_{\mathclap{y \in \cube n}} \kb{y} \otimes \kb{0^{m-n}} \otimes \kb1 \right) \left( \adj A \otimes I_1 \right) \\
		&= I_{m+1} - 2\sum_{\mathclap{y \in \cube n}} \kb{y} \otimes U \kb{y} \adj U \otimes \kb{0^{m-2n}} \otimes \kb1 \\
		&= I_{m+1} - 2\sum_{\mathclap{y \in \cube n}} \kb{y, 1} \otimes U \kb{y} \adj U \otimes \kb{0^{m-2n}},
	\end{align*}
	where in the last line we reorder the qubits for future convenience. Therefore
	\begin{equation*}
		C \Paren{I_{n+1} \otimes U \ket{x} \otimes \ket{0^{m-2n}}}
		= (I_{n+1} - 2\kb{x,1}) \otimes U\ket x \otimes \ket{0^{m-2n}},
	\end{equation*}
	so using our copy of $U\ket{x}$, the circuit $C$ can implement the reflection $I_{n+1} - 2\kb{x,1}$ in a disjoint register without disturbing the copy of $U\ket{x}$.
	
	We could therefore simulate the circuit $G_n$ from \cref{egs}, with queries to $I - 2\kb{x,1}$ answered in this manner, to construct a copy of $\ket x$. Instead perform this simulation \emph{in reverse}, to \emph{uncompute} the existing copy of $\ket{x}$ while preserving the copy of $U\ket{x}$. Finally swap $U\ket{x}$ into the appropriate register.
\end{proof}

\subsection{Lower bound} \label{s:lb}

For linear transformations $L, M$ from $n$ qubits to $m$ qubits where $n \le m$ let $\fip L M = 2^{-n} \tr{\adj L M}$, i.e.\ $\fip \cdot \cdot$ is the Frobenius inner product normalized such that $\fip A A = 1$ for all isometries $A$.

\begin{thm}[formal version of \cref{sc-lb}] \label{lb}
	For all sequences of quantum circuits $(C_n)_n$ making $o\Paren{2^{n/2}}$ queries to a $2n$-qubit quantum oracle, with probability $1-o(1)$ over a Haar random $n$-qubit unitary $U$, there exists a $2n$-qubit $U$-CC $A$ such that for all states $\ket\psi$,
	\begin{equation*}
		\Mag{\fip {C_n^A (I_n \otimes \ket\zs)} {U \otimes \ket\psi}} \le o(1).
	\end{equation*}
\end{thm}

\begin{proof}
	For a permutation $\sigma$ of $\cube n$ let $A_\sigma$ be the unitary defined by $A_\sigma \ket{x, y} = \ket{x, y \oplus \sigma(x)}$ for all $x,y \in \cube n$. Nayak~\cite[Corollary 1.2]{Nay11} proved that any quantum circuit making $o \Paren{2^{n/2}}$ queries to $A_\sigma$ outputs $\sigma^{-1} \Paren{0^n}$ with probability less than $1/2$, where the probability is over a uniform random permutation $\sigma$ of $\cube n$ as well as the randomness of the output measurement. (We remark that Ambainis~\cite{Amb02} previously proved a similar result using different techniques.)
	
	Let $\eps, \delta > 0$ be universal constants, and assume for the sake of contradiction that there exists a quantum circuit $C$ making $o\Paren{2^{n/2}}$ queries to a $2n$-qubit quantum oracle, such that with probability at least $\eps$ over a Haar random $n$-qubit unitary $U$, for all $2n$-qubit $U$-CCs $A$, there exists a state $\ket\psi$ such that $\Mag{\fip {C^A (I_n \otimes \ket\zs)} {U \otimes \ket\psi}} \ge \delta$. We will prove that there exists a quantum circuit making $o\Paren{2^{n/2}}$ queries to $A_\sigma$ that outputs $\sigma^{-1} \Paren{0^n}$ with probability $\Omega(1)$, where the probability is over a uniform random permutation $\sigma$ of $\cube n$ as well as the randomness of the output measurement. By executing this circuit constantly many times until it outputs $\sigma^{-1} \Paren{0^n}$, we can boost the success probability to be greater than 1/2 which contradicts Nayak's result. Therefore no such circuit $C$ exists.
	
	Write $C = C_s Q_s C_{s-1} Q_{s-1} \dotsb C_0$ for $s = o\Paren{2^{n/2}}$, where each $C_i$ is a unitary and each $Q_i$ is a placeholder for either a forward or backward query. For an $n$-qubit unitary $R$, define a quantum circuit $C_R$ by replacing each forward query $Q_i$ in $C$ with $(I_n \otimes R) Q_i$, and replacing each backward query $Q_i$ in $C$ with $Q_i \Paren{I_n \otimes \adj R}$. For a permutation $\sigma$ of $\cube n$ let $P_\sigma$ denote the corresponding permutation matrix on $n$ qubits, i.e.\ $P_\sigma \ket{x} = \ket{\sigma(x)}$ for all $x \in \cube n$. Clearly for all $R,\sigma$ it holds that $C_R^{A_\sigma} = C^{(I_n \otimes R) A_\sigma}$, and that $(I_n \otimes R) A_\sigma$ is a $2n$-qubit $RP_\sigma$-CC. If $\sigma$ is fixed and $R$ is Haar random, then $RP_\sigma$ is also Haar random, so for some state $\ket\psi$ depending on $R$ and $\sigma$,
	\begin{equation*}
		\PRs R {\Mag{\fip { C_R^{A_\sigma} (I_n \otimes \ket\zs)} {R P_\sigma \otimes \ket\psi}} \ge \delta} \ge \eps.
	\end{equation*}
	
	Call a fixed unitary $R$ ``good with respect to $\sigma$" if $\Mag{\fip { C_R^{A_\sigma} (I_n \otimes \ket\zs)} {R P_\sigma \otimes \ket\psi}} \ge \delta$ for some $\ket\psi$. Also let $D_R = \adj C_R (R \otimes \ket\psi)$. (For intuition, if $R$ is good with respect to $\sigma$ then $C_R^{A_\sigma}$ approximately implements $RP_\sigma$, and so $D_R$ approximately implements $P_{\sigma^{-1}}$.) If $R$ is good with respect to $\sigma$ then
	\begin{align*}
		\delta
		&\le \Mag{2^{-n} \tr{\left( I_n \otimes \bra\zs \right) \adj {\Paren{C_R^{A_\sigma}}} \left(RP_\sigma \otimes \ket\psi\right)}} &\text{(definition of $\fip \cdot \cdot$)} \\
		&= \Mag{2^{-n} \sum_{\mathclap{x \in \cube n}} \bra{x,\zs} \adj {\Paren{C_R^{A_\sigma}}} \left(RP_\sigma \ket{x} \otimes \ket\psi \right)} &\text{(definition of trace)} \\
		&\le 2^{-n} \sum_{\mathclap{x \in \cube n}} \left| \bra{x,\zs} D_R^{A_\sigma} \ket{\sigma(x)} \right| &\text{(triangle ineq., definitions of $D_R, P_\sigma$)} \\
		&\le 2^{-n} \sum_{\mathclap{x \in \cube n}} \Norm{\Paren{\bra{\sigma^{-1}(x)} \otimes I} D_R^{A_\sigma} \ket{x}} &\text{(Cauchy-Schwarz, $x \gets \sigma^{-1}(x)$)}.
	\end{align*}
	For $x \in \cube n$ let $p_{\sigma, R, x}$ be the probability that if we run $D_R^{A_\sigma}$ on input $x$ and measure the first $n$ qubits of the output state, then the result is $\sigma^{-1}(x)$. Then we can phrase the above inequality as $\delta \le 2^{-n} \sum_{x \in \cube n} \sqrt{p_{\sigma, R, x}}$, and by Cauchy-Schwarz it follows that $\delta^2 \le 2^{-n} \sum_{x \in \cube n} p_{\sigma, R, x}$.
	
	Therefore for every fixed permutation $\sigma$ of $\cube n$, for Haar random $R$ and uniform random $x \in \cube n$, it holds that
	\begin{equation*}
		\eps \le \underset{R}\prz \left( \text{$R$ is good w.r.t.\ $\sigma$} \right)
		\le \underset{R}\prz \left( \delta^2 \le \underset{x} \E \left[ p_{\sigma, R, x} \right] \right)
		\le \delta^{-2} \, \underset{\mathclap{R,x}}\E \left[ p_{\sigma, R, x} \right]
	\end{equation*}
	where the last step is by Markov's inequality. If we also take $\sigma$ to be uniform random then $\E_{\sigma, R, x} \left[ p_{\sigma, R, x} \right] \ge \eps \delta^2$, so there exist \emph{fixed} values of $R$ and $x$ such that $\E_\sigma \left[ p_{\sigma, R, x} \right] \ge \eps \delta^2$. Thus there exists a quantum circuit (specifically $D_R^{A_\sigma} \ket x$ for these fixed values of $R$ and $x$) making $o \left( 2^{n/2} \right)$ queries to $A_\sigma$ that outputs $\sigma^{-1}(x)$ with probability at least $\eps \delta^2$, where the probability is over a uniform random permutation $\sigma$ of $\cube n$ as well as the randomness of the output measurement. By symmetry such a circuit exists with $x = 0^n$ as desired.
\end{proof}

\section{Upper bound for the unitary synthesis problem} \label{s:class}

We prove the following:

\begin{thm}[formal version of \cref{class}] \label{thm:usp}
	Let $\eps(n) = \exp(-\poly(n))$. Then there is a uniform sequence of \qacf circuits $(C_n)_n$---each of depth $\oo{2^{n/2}}$, making $\oo{2^{n/2}}$ queries, and with $\poly(n)$ ancillae---such that for all $n$-qubit unitaries $U$ there exists a classical oracle $f$ such that $\Norm{C_n^f (I_n \otimes \ket\zs) - U \otimes \ket\zs} \le \eps(n)$.
\end{thm}

Although \cref{thm:usp} is stated in terms of \qacf circuits for convenience, a similar statement for \qnc circuits follows easily using \cref{lem:qsim}. A similar statement also holds for diamond norm error: For an isometry $A$ define a channel $\chan A$ by $\chan A(\rho) = A \rho \adj A$. For all states $\ket\psi$ and $\ket\phi$ it holds that $\frac12 \Norm{\kb\psi - \kb\phi}_1 \le \Norm{\ket\psi - \ket\phi}$~\cite[Eq.\ 2.3]{Ros24}, so for all isometries $A$ and $B$ it holds that
\begin{align}
	\frac12 \Norm{\chan A - \chan B}_\diamond
	&= \frac12 \max_{\ket\psi} \Norm{((\chan A - \chan B) \otimes I) (\kb\psi)}_1 \notag\\
	&= \frac12 \max_{\ket\psi} \Norm{(A \otimes I) \kb\psi \Paren{\adj A \otimes I} - (B \otimes I) \kb\psi \Paren{\adj B \otimes I}}_1 \notag\\
	&\le \max_{\ket\psi} \Norm{(A \otimes I)\ket\psi - (B \otimes I)\ket\psi} \notag\\
	&= \Norm{(A - B) \otimes I} \notag\\
	&= \norm{A - B}. \label{eq:dist-comp}
\end{align}
In particular, if we define a channel $\chan n^f$ by tracing out the ancillae after applying $C_n^f (I_n \otimes \ket\zs)$, then
\begin{equation*}
	\frac12 \Norm{\chan n^f - \chan U}_\diamond
	\le \frac12 \Norm{\chan{C_n^f (I_n \otimes \ket\zs)} - \chan{U \otimes \ket\zs}}_\diamond
	\le \eps(n),
\end{equation*}
where the first inequality holds because tracing out qubits at the end cannot increase the diamond norm distance between two channels, and the second inequality is by \cref{eq:dist-comp,thm:usp}.

Queries to a classical oracle (i.e.\ a Boolean function) can be modeled in either of two standard ways. In the first, a function $f: \cube n \to \cube m$ is encoded as the oracle $U_f$ defined by $U_f \ket{x, y} = \ket{x, y \oplus f(x)}$. In the second, which is only applicable when $m=1$, the function $f$ is instead encoded as the oracle $V_f$ defined by $V_f \ket{x} = (-1)^{f(x)} \ket{x}$. These models are equivalent, because $V_f = (I_n \otimes \bra-) U_f (I_n \otimes \ket-)$ where $\ket- = \frac{\ket0 - \ket1}{\sqrt2}$, and if $g(x,y) = \bigoplus_{j=1}^m f(x)_j y_j$ (where the subscript $j$ indicates the $j$-th bit of an $m$-bit string) then $U_f = (I_n \otimes H^{\otimes m}) V_g (I_n \otimes H^{\otimes m})$ where $H$ denotes the Hadamard gate. We write $C^f$ to abbreviate $C^{U_f}$ or $C^{V_f}$, as defined in the beginning of \cref{sec:qram}; since $U_f$ and $V_f$ are Hermitian we do not need to distinguish between forward and backward queries.

Our proof uses the following result:\footnote{An earlier version of our proof used a similar result of Aaronson~\cite[Proposition 3.3.5]{Aar16} that required $O(n)$ queries, resulting in a multiplicative $O(n)$ blowup in the query complexity in \cref{thm:usp}.}

\begin{thm}[{Rosenthal~\cite{Ros24}}] \label{thm:fqa}
	Let $\eps(n) = \exp(-\poly(n))$. Then there is a uniform sequence of $\poly(n)$-qubit \qaczf circuits $(C_n)_n$, each making four queries, such that for all $n$-qubit states $\ket\psi$ there exists a classical oracle $f$ such that $\Norm{C_n^f \ket\zs - \ket\psi \ket\zs} \le \eps(n)$.
\end{thm}

Our proof also uses the following lemma, which says that if the task of implementing an isometry $J$ reduces to that of implementing an isometry $A$, then the task of \emph{approximately} implementing $J$ reduces to that of \emph{approximately} implementing $A$:

\begin{lem} \label{lem:err}
	Let $C$ be an $(m+a)$-qubit quantum circuit making $k$ queries to an $n$-qubit quantum oracle, and let $J$ be an isometry from $m$ qubits to $m+a$ qubits. Assume there exists a subspace $S \subseteq \Paren{\C^2}^{\otimes n}$ and an isometry $A: S \to \Paren{\C^2}^{\otimes n}$ such that for all $n$-qubit unitaries $U$ consistent with $A$ it holds that $C^U (I_m \otimes \ket{0^a}) = J$. Then for all isometries $B: S \to \Paren{\C^2}^{\otimes n}$ and all $n$-qubit unitaries $V$ consistent with $B$, it holds that $\Norm{C^{V} (I_m \otimes \ket{0^a}) - J} \le \sqrt 2 \cdot k \norm{A - B}$.
\end{lem}

First we prove \cref{thm:usp} assuming \cref{lem:err}, and then we prove \cref{lem:err}.

\begin{proof}[Proof of \cref{thm:usp}]
	\cref{thm:fqa} generalizes from constructing states to implementing $\poly(n)$-qubit $U$-CCs (for an $n$-qubit unitary $U$), because if $f_x$ is the oracle associated with constructing $U \ket x$ in \cref{thm:fqa}, then the function $(x,y) \mapsto f_x(y)$ can simulate queries to $f_x$ controlled on $x$. Therefore there is a uniform sequence of $\poly(n)$-qubit \qaczf circuits $(A_n)_n$, each making four queries, such that for all $n$-qubit unitaries $U$ there exists a classical oracle $f$ such that
	\begin{equation*}
		\max_{x \in \cube n} \Norm{A_n^f \ket{x, 0^{n+m}} - \ket{x} \otimes U \ket{x} \otimes \ket{0^m}} \le \eps(n) / \Paren{c2^n \cdot \sqrt 2}
	\end{equation*}
	for $m = \poly(n)$. Here $c$ is a constant such that the circuit in \cref{main} makes at most $c 2^{n/2}$ queries. Since the operator norm of a matrix is at most the Frobenius norm, it follows that
	\begin{align*}
		\Norm{A_n^f \Paren{I_n \otimes \ket{0^{n+m}}} - \sum_{\mathclap{x \in \cube n}} \kb{x} \otimes U \ket{x} \otimes \ket{0^m}}
		&\le \sqrt{\sum_{x \in \cube n} \Norm{A_n^f \ket{x, 0^{n+m}} - \ket{x} \otimes U \ket{x} \otimes \ket{0^m}}^2} \\
		&\le \eps(n) / \Paren{c 2^{n/2} \cdot \sqrt 2}.
	\end{align*}
	The result then follows by \cref{main,lem:err}, the latter applied with isometries $J = U \otimes \ket\zs$ and $A$ a $U$-CC and $B = A_n^f \Paren{I_n \otimes \ket{0^{n+m}}}$.
\end{proof}
\subsection{Proof of \texorpdfstring{\cref{lem:err}}{Lemma 3.3}}

We use the fact that
\begin{equation} \label{eq:nc-err}
	\Norm{U_m U_{m-1} \dotsb U_1 - V_m V_{m-1} \dotsb V_1} \le \sum_{j=1}^m \Norm{U_j - V_j}
\end{equation}
for all unitaries $U_j, V_j$ by the triangle inequality~\cite[Eq.\ 4.69]{NC10}. (The reason that \cref{lem:err} does not trivially follow from \cref{eq:nc-err}, and without the $\sqrt2$ factor, is that some of the states acted on by applications of $V$ in $C^V (I_m \otimes \ket{0^a})$ might be far from $S$.)

For a unitary $U \in \C^{n \times n}$ and a subspace $S \subseteq \C^n$, we write $U_{\mid S}$ to denote the isometry from $S$ to $\C^n$ defined by restricting the domain of $U$ to $S$. The following inequality is trivially tight to within a $\sqrt2$ factor:

\begin{clm} \label{clm:err}
	For all unitaries $V \in \C^{n \times n}$, subspaces $S \subseteq \C^n$, and isometries $W: S \to \C^n$, there exists a unitary $U \in \C^{n \times n}$ such that $U_{\mid S} = W$ and $\norm{U - V} \le \sqrt 2 \Norm{W - V_{\mid S}}$.
\end{clm}

\cref{lem:err} follows immediately from \cref{eq:nc-err,clm:err}, along with the fact that $\Norm{\adj U - \adj V} = \norm{U - V}$ (to handle backward queries).

\begin{proof}[Proof of \cref{clm:err}]
	Let $P \Sigma \adj Q$ be a singular value decomposition of $\adj W V_{\mid S}$, and write $\Sigma = \mrm{diag} \Paren{\sigma_1, \dotsc, \sigma_k}$ where $k = \dim(S)$ and $\sigma_1 \ge \dotsb \ge \sigma_k \ge 0$. By definition $P, Q: \C^k \to S$ are (bijective) isometries such that $P \Sigma \adj Q = \adj W V_{\mid S}$. Then
	\begin{align*}
		\Norm{W - V_{\mid S}}
		&\ge \Norm{\Paren{W - V_{\mid S}} Q \ket k}
		\ge \sqrt{2 - 2 \Mag{\bra{k} \adj Q \adj W V_{\mid S} Q \ket k}}
		= \sqrt{2 - 2 \sigma_k \Mag{\bra{k} \adj Q P \ket k}} \\
		&\ge \sqrt{2 - 2\sigma_k},
	\end{align*}
	where the last inequality is by Cauchy-Schwarz.
	
	Define isometries $A, B \in \C^{n \times k}$ by $A = WP$ and $B = V_{\mid S} Q$, and note that $\adj A B = \adj P \adj W V_{\mid S} Q = \Sigma$. Let
	\begin{equation*}
		\overline A = \mrm{trnc}\Brac{(B - A \Sigma) \Paren{I - \Sigma^2}^{-1/2}}, \qquad
		\overline B = \mrm{trnc}\Brac{(B \Sigma - A) \Paren{I - \Sigma^2}^{-1/2}},
	\end{equation*}
	where the ``truncation" operator $\mrm{trnc}[\cdot]$ removes the $j$'th column of a matrix for all $j$ such that $\sigma_j = 1$ (thus avoiding division by zero). It is straightforward to verify that $\overline A, \overline B$ are isometries satisfying $\adj A \overline A = \adj B \overline B = 0$ and $\adj{\overline A} \overline B = \mathrm{trnc} [\Sigma]$, so
	\begin{equation*}
		\Norm{\overline A - \overline B} = \max_{\Norm{\ket\psi} = 1} \Norm{\overline A \ket\psi - \overline B \ket\psi} = \max_{\mathclap{\Norm{\ket\psi} = 1}} \sqrt{2 - 2 \bra\psi \mathrm{trnc} [\Sigma] \ket\psi} = \sqrt{2 - 2\sigma_k} \le \Norm{W - V_{\mid S}}.
	\end{equation*}
	
	Let $S^\prime \subseteq \C^n$ be the subspace that $V$ maps to the image of $\overline B$. The subspaces $S$ and $S^\prime$ are orthogonal, because $\adj B \overline B = 0$ and $V$ maps $S$ to the image of $B$, so we can write $\C^n = S \oplus S^\prime \oplus S^{\prime \prime}$ for some subspace $S^{\prime \prime}$. Let $R: S^\prime \to \C^{\Mag{\{j:\, \sigma_j \neq 1\}}}$ be the (bijective) isometry such that $V_{\mid S^\prime} = \overline B R$, and define a unitary $U \in \C^{n \times n}$ by
	\begin{equation*}
		U_{\mid S} = W, \qquad
		U_{\mid S^\prime} = \overline A R, \qquad
		U_{\mid S^{\prime \prime}} = V_{\mid S^{\prime \prime}}.
	\end{equation*}
	To see that $U$ is in fact unitary, note that $\adj W \cdot \overline A R = P \adj A \overline A R = 0$, and that the image of $U_{\mid S \oplus S^\prime}$ (i.e.\ the span of $A$ and $\overline A$) equals the image of $V_{\mid S \oplus S^\prime}$ (i.e.\ the span of $B$ and $\overline B$). The latter condition holds because for all columns $j$, if $\sigma_j \neq 1$ then $\mrm{span}(A \ket j, \overline A \ket j) = \mrm{span}(A \ket j, B \ket j) = \mrm{span}(B \ket j, \overline B \ket j)$, and if $\sigma_j = 1$ then $A \ket j = B \ket j$ (recalling that $\adj A B = \Sigma$) and $\overline A \ket j, \overline B \ket j$ have not been defined (recalling the definition of $\mrm{trnc}[\cdot]$).
	
	Let $\ket\psi \in \C^n$ be a unit vector such that $\Norm{U - V} = \Norm{(U - V) \ket\psi}$, and for a subspace $T \subseteq \C^n$ let $\ket{\psi_T}$ be the projection of $\ket\psi$ onto $T$. Then by the triangle inequality and Cauchy-Schwarz,
	\begin{align*}
		\Norm{U - V}
		&= \Norm{(U - V) \ket\psi}
		\le \sum_{\mathclap{T \in \{S, S^\prime, S^{\prime \prime}\}}} \Norm{(U - V)_{\mid T} \ket{\psi_T}}
		\le \Norm{W - V_{\mid S}} \cdot \Norm{\ket{\psi_S}} + \Norm{\overline A - \overline B} \cdot \Norm{\ket{\psi_{S^\prime}}} \\
		&\le \Norm{W - V_{\mid S}} \Paren{\Norm{\ket{\psi_S}} + \Norm{\ket{\psi_{S^\prime}}}}
		\le \sqrt 2 \Norm{W - V_{\mid S}} \cdot \Norm{\ket{\psi_{S \oplus S^\prime}}}
		\le \sqrt 2 \Norm{W - V_{\mid S}}. \qedhere
	\end{align*}
\end{proof}

\color{black}

\section{Circuit depth upper bounds for states and unitaries} \label{sec:depth}
In \cref{s:state-qacf,s:u-qacf} respectively we prove circuit depth upper bounds for constructing arbitrary states and implementing arbitrary unitaries.
\subsection{States} \label{s:state-qacf}

We prove the following:

\begin{thm}[formal version of \cref{state-qacf}] \label{thm:qacf-formal}
	For all $n$-qubit states $\ket\psi$ there exists an $\ot{2^n}$-qubit \qaczf circuit $C$ such that $C \ket\zs = \ket\psi \ket\zs$.
\end{thm}

\cref{thm:qacf-formal,lem:qsim} imply the following:

\begin{cor}[formal version of \cref{state-qnc}]
	For all $n$-qubit states $\ket\psi$ there exists an $O(n)$-depth, $\ot{2^n}$-qubit \qnc circuit $C$ such that $C \ket\zs = \ket\psi \ket\zs$.
\end{cor}

A proof sketch of \cref{thm:qacf-formal} is as follows. First consider the analogous problem of sampling a string $s$ from a given distribution over $\cube n$. One way to sample $s$ is to first sample
\begin{equation*}
	b_x \sim \mathrm{Bernoulli} \left( \pr{\text{$s$ begins with $x1$} \mid \text{$s$ begins with $x$}} \right)
\end{equation*}
independently for all binary strings $x$ of length less than $n$, and then output the string $y$ defined by $y_i = b_{y_1 y_2 \dotsb y_{i-1}}$ for $i$ from $1$ to $n$. Furthermore each bit of $y$ can be computed by a DNF formula of size $\ot{2^n}$ as a function of $(b_x)_x$. Similarly we can construct a quantum state $\sum_{y \in \cube n} \alpha_y \ket{y}$ using unentangled one-qubit states in place of $(b_x)_x$; this actually yields a state of the form $\sum_{y \in \cube n} \alpha_y \ket{y} \ket{\mathrm{garbage}_y}$, but it turns out that $\ket{\mathrm{garbage}_y}$ can be efficiently uncomputed controlled on $y$.

Our proof will use the following notation. Let $\cube{\le n}$ (resp.\ $\cube{<n}$) denote the set of strings of length at most (resp.\ less than) $n$ over $\bits$, including the empty string $\es$. For $x \in \cube*$ let $x_k, x_{<k}, x_{\le k}$ respectively denote the $k$-th bit, first $k-1$ bits, and first $k$ bits of $x$, and let $|x|$ denote the length of $x$. For $x,y \in \cube*$ let $xy$ denote the concatenation of $x$ and $y$.

\begin{proof}
	Let $\ket\psi = \sum_{x \in \cube n} \alpha_x \ket x$ denote the $n$-qubit state to be constructed, and define ``conditional amplitudes" $\beta_x$ for $x \in \cube{\le n} \backslash \{\es\}$ as follows: Let $\ket{\psi_\es} = \ket\psi$, and for $x \in \cube{< n}$, given an $(n-|x|)$-qubit state $\ket{\psi_x}$, write
	\begin{equation*}
		\ket{\psi_x} =
		\begin{cases}
			\beta_{x0} \ket0 \ket{\psi_{x0}} + \beta_{x1} \ket1 \ket{\psi_{x1}} & \text{if } |x| \le n-2, \\
			\beta_{x0} \ket0 + \beta_{x1} \ket1 & \text{if } |x| = n-1
		\end{cases}
	\end{equation*}
	for $(n-|x|-1)$-qubit states $\ket{\psi_{x0}}, \ket{\psi_{x1}}$ (if $|x| \le n-2$) and complex numbers $\beta_{x0}, \beta_{x1}$ such that $\Mag{\beta_{x0}}^2 + \Mag{\beta_{x1}}^2 = 1$. Let
	\begin{equation*}
		\ket{\phi_x} = \beta_{x0} \ket0 + \beta_{x1} \ket1
	\end{equation*}
	for $x \in \cube{<n}$, and observe that $\alpha_x = \prod_{i=1}^n \beta_{x_{\le i}}$ for all $x \in \cube n$.
	
	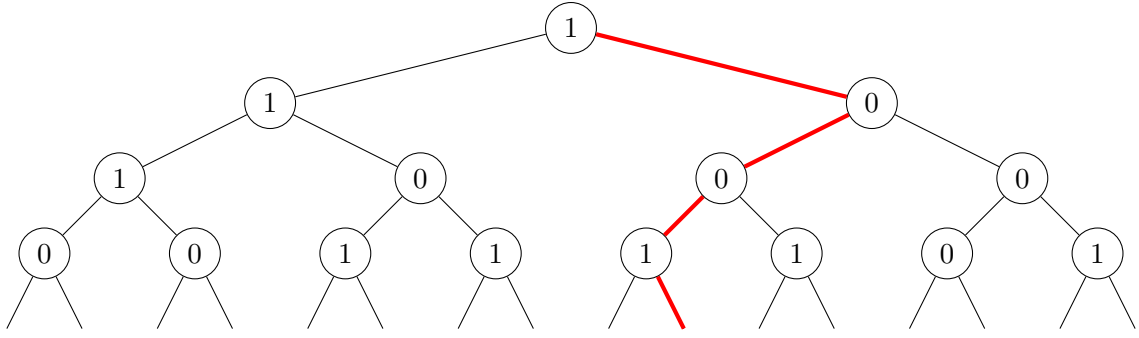
\begin{figure}
		\centering
		\begin{tikzpicture}[level distance=1cm,
			level 1/.style={sibling distance=8cm},
			level 2/.style={sibling distance=4cm},
			level 3/.style={sibling distance=2cm},
			level 4/.style={sibling distance=1cm},
			highlight/.style={ultra thick, red}]
			
			\node[circle, draw] (root) {1}
			child {
				node[circle, draw] {1}
				child {
					node[circle, draw] {1}
					child {
						node[circle, draw] {0}
						child[draw=none] {}
						child[draw=none] {}
					}
					child {
						node[circle, draw] {0}
						child[draw=none] {}
						child[draw=none] {}
					}
				}
				child {
					node[circle, draw] {0}
					child {
						node[circle, draw] {1}
						child[draw=none] {}
						child[draw=none] {}
					}
					child {
						node[circle, draw] {1}
						child[draw=none] {}
						child[draw=none] {}
					}
				}
			}
			child {
				node[circle, draw] {0}
				child {
					node[circle, draw] {0}
					child {
						node[circle, draw] {1}
						child[draw=none] {}
						child[draw=none] {}
					}
					child {
						node[circle, draw] {1}
						child[draw=none] {}
						child[draw=none] {}
					}
				}
				child {
					node[circle, draw] {0}
					child {
						node[circle, draw] {0}
						child[draw=none] {}
						child[draw=none] {}
					}
					child {
						node[circle, draw] {1}
						child[draw=none] {}
						child[draw=none] {}
					}
				}
			};
			
			\draw[highlight] (root) -- (root-2);
			\draw[highlight] (root-2) -- (root-2-1);
			\draw[highlight] (root-2-1) -- (root-2-1-1);
			\draw[highlight] (root-2-1-1) -- (root-2-1-1-2);
		\end{tikzpicture}
		\caption{Nodes are labeled with the inputs to $f$. The highlighted path represents the output of $f$, and is defined by starting at the root and repeatedly walking to the left or right child depending on whether the current node is labeled 0 or 1.}
		\label{fig:tree}
	\end{figure}
	
	Let $f : \cube{\cube{<n}} \to \cube n$ be the function defined by $f(x)_i = x_{f(x)_{<i}}$ for $i$ from 1 to $n$. The function $f$ is illustrated in \cref{fig:tree} and can be computed as follows:
	\begin{equation*}
		f(x)_j = \bigvee_{\substack{t \in \cube j \\ t_j = 1}} \bigwedge_{1 \le i \le j} \Ind{x_{t_{<i}} = t_i} \qquad \text{for $1 \le j \le n$}.
	\end{equation*}
	(The conjunction indicates whether $t$ equals the first $j$ bits of $f(x)$, and the disjunction indicates whether the satisfying $t$ is such that $t_j = 1$.) Each of the $n$ \acz formulas has $\sum_{j=1}^n 2^{j-1} j \le \ot{2^n}$ leafs, so a \qaczf circuit on $\ot{2^n}$ qubits can compute the unitary
	\begin{equation*}
		U_f \ket{x, a} = \ket{x, a \oplus f(x)} \qquad \text{for $x \in \cube{\cube{<n}}, a \in \cube{n}$}
	\end{equation*}
	by first fanning out the input bits to make as many copies as needed, then simulating the gates in the \acz formulas, and finally uncomputing the garbage.
	
	Let $(\reg R_x)_{x \in \cube{<n}}$ be one-qubit registers and let $\reg S$ be an $n$-qubit register. The first step toward constructing $\ket\psi$ is to construct the state
	\begin{equation*}
		U_f \left(\quad \bigotimes_{\mathclap{x \in \cube{<n}}} \ket{\phi_x}_{\reg R_x} \otimes \ket{0^n}_{\reg S} \right),
	\end{equation*}
	using a layer of one-qubit gates followed by the aforementioned circuit for $U_f$. Here, when computing $U_f$, the $x$-th input bit to $f$ is in $\reg R_x$ for all $x$, and the output register of $f$ is $\reg S$. Observe that
	\begin{align*}
		U_f (I \otimes \ket{0^n}) = \sum_{\mathclap{x \in \cube{\cube{<n}}}} \kb{x} \otimes \ket{f(x)}
		= \sum_{\mathclap{t \in \cube n}} \,\,\, \left(\quad \sum_{\mathclap{x \in f^{-1}(t)}} \kb{x} \right) \otimes \ket{t}
		= \sum_{\mathclap{t \in \cube n}} \,\,\, \left( \bigotimes_{i=1}^n \kb{t_i}_{\reg R_{t_{<i}}} \right) \otimes \ket{t}_{\reg S},
	\end{align*}
	where the $t$-th tensor product above implicitly acts as the identity on all $\reg R_x$ for which $x$ does not equal $t_{<i}$ for any $i$. Therefore
	\begin{equation*}
		U_f \left(\quad \bigotimes_{\mathclap{x \in \cube{<n}}} \ket{\phi_x}_{\reg R_x} \otimes \ket{0^n}_{\reg S} \right)
		= \sum_{t \in \cube n} \bigotimes_{x \in \cube{<n}} \begin{cases} \ket{t_i} \ip{t_i}{\phi_{t_{<i}}}_{\reg R_x} & \text{if $x = t_{<i}$ for some $i$} \\ \ket{\phi_x}_{\reg R_x} & \text{otherwise} \end{cases} \otimes \ket{t}_{\reg S}.
	\end{equation*}
	By the definition of $\ket{\phi_{t_{<i}}}$ it holds that $\ip {t_i} {\phi_{t_{<i}}} = \beta_{t_{<i} t_i} = \beta_{t_{\le i}}$, so since $\alpha_t = \prod_{i=1}^n \beta_{t_{\le i}}$ for all $t \in \cube n$ it follows that
	\begin{equation*}
		U_f \left(\quad \bigotimes_{\mathclap{x \in \cube{<n}}} \ket{\phi_x}_{\reg R_x} \otimes \ket{0^n}_{\reg S} \right)
		= \sum_{t \in \cube n} \alpha_t \bigotimes_{x \in \cube{<n}} \begin{cases} \ket{t_i}_{\reg R_x} & \text{if $x = t_{<i}$ for some $i$} \\ \ket{\phi_x}_{\reg R_x} & \text{otherwise} \end{cases} \otimes \ket{t}_{\reg S}.
	\end{equation*}
	
	All that remains to construct the state $\ket\psi = \sum_{t \in \cube n} \alpha_t \ket t$ is to uncompute the above content of $(\reg R_x)_{x \in \cube{<n}}$ controlled on the state $\ket t$ of $\reg S$. To do so, first make $\left| \cube{<n} \right|$ copies of $t$ using fanout. Then for each $x \in \cube{<n}$ in parallel, controlled on one of these copies of $t$, if $x = t_{<i}$ for some $i$ then perform in $\reg R_x$ an operation that maps $\ket{t_i}$ to $\ket0$, and otherwise perform in $\reg R_x$ an operation that maps $\ket{\phi_x}$ to $\ket 0$. Finally, uncompute the extra copies of $t$ using fanout.
\end{proof}

\subsection{Unitaries} \label{s:u-qacf}

First we establish some basic properties of \qacf circuits:

\begin{lem} \label{lem:qaczf-swap}
	There is a uniform family of $O\Paren{m n \log n}$-qubit \qaczf circuits $(C_{n,m})_{n,m}$, where $C_{n,m}$ takes as input a $(\log n)$-qubit register $\reg K$ and $m$-qubit registers $\reg A_0, \dotsc, \reg A_{n-1}, \reg B$ (and ancillae) and $C_{n,m}$ swaps $\reg A_k$ and $\reg B$ controlled on the classical state $\ket{k}_{\reg K}$.
\end{lem}
\begin{proof}
	We can assume without loss of generality that $m=1$, because then the general case follows by swapping the $i$-th qubits of $\reg A_k$ and $\reg B$ for all $i$ in parallel. By linearity we may assume that the input is a standard basis state
	\begin{equation*}
		\ket{k}_{\reg K} \ket{x_0}_{\reg A_0} \dotsb \ket{x_{n-1}}_{\reg A_{n-1}} \ket{y}_{\reg B}.
	\end{equation*}
	
	For now assume that $y$ is promised to be $0$. First compute $x_k = \bigvee_{j=0}^{n-1} \Paren{\Ind{j=k} \wedge x_j}$ in $\reg B$, using that \qaczf circuits can simulate \acz circuits; note that comparing $j$ and $k$ requires $O(\log n)$ qubits for any given value of $j$. Then controlled on the state $\ket{x_k}_{\reg B}$, make $n$ copies of $x_k$ with fanout, use them to XOR the bit $\Ind{j=k} \wedge x_k$ into $\reg A_j$ for all $j < n$ in parallel, and finally uncompute the extra copies of $x_k$ using fanout.
	
	For the general case where $y$ might not be $0$, let $\reg C$ be a one-qubit register in the ancillae. First swap $\reg A_k$ and $\reg C$ as described above, then swap $\reg B$ and $\reg A_k$ as described above, and finally swap $\reg C$ and $\reg B$.
\end{proof}

\begin{lem} \label{lem:cqacz}
	If $C$ is an $n$-qubit, size-$s$, depth-$d$ \qacf circuit then controlled-$C$ can be implemented by an $O(n)$-qubit, size-$O(s)$, depth-$O(d)$ \qacf circuit.
\end{lem}
\begin{proof}
	Controlled on a bit $b \in \bits$, each gate in a \qacf circuit can be implemented controlled on $b$ as follows. A $k$-qubit generalized Toffoli gate controlled on $b$ is equivalent to a $(k+1)$-qubit generalized Toffoli gate, fanning out a bit $c$ controlled on $b$ is equivalent to fanning out $bc$, and applying a one-qubit gate controlled on $b$ is a two-qubit operation and can therefore be implemented with a constant number of one-qubit and CNOT gates~\cite[Section 4.5.2]{NC10}. The result follows by making $n$ copies of $b$, and using these copies to implement all gates in a given layer of $C$ in parallel controlled on $b$, where the same ancillae are reused in simulations of successive layers of $C$.
\end{proof}

Now we prove that $\oo{2^{n/2}}$-depth \qacf circuits can implement any $n$-qubit unitary:

\begin{thm} \label{thm:unit-qacf}
	For all $n$-qubit unitaries $U$ there exists an $\oo{2^{n/2}}$-depth, $\ot{2^{2n}}$-qubit \qacf circuit $C$ such that $C(I_n \otimes \ket\zs) = U \otimes \ket\zs$.
\end{thm}
\begin{proof}
	By \cref{main} it suffices to implement a $U$-CC with an $\ot{2^{2n}}$-qubit \qaczf circuit, and this can be achieved as follows. On input $x \in \cube n$ to the $U$-CC, first make $2^n$ copies of $x$ using fanout. Then for all $y \in \cube n$ in parallel, in a register $\reg R_y$ use \cref{thm:qacf-formal,lem:cqacz} to construct $U \ket y$ controlled on $x = y$, using a different copy of $x$ for each string $y$. Then swap $\reg R_x$ into the output register using \cref{lem:qaczf-swap}, and use fanout to uncompute the extra copies of $x$.
\end{proof}

\cref{thm:unit-qacf,lem:qsim} imply the following:

\begin{cor}[formal version of \cref{depth}]
	For all $n$-qubit unitaries $U$ there exists an $\ot{2^{n/2}}$-depth, $\ot{2^{2n}}$-qubit \qnc circuit $C$ such that $C(I_n \otimes \ket\zs) = U \otimes \ket\zs$.
\end{cor}

\section*{Acknowledgments}

Thanks to Scott Aaronson, Karen J.\ Morenz Korol, Fermi Ma, Adrian She, Nathan Wiebe, and Henry Yuen for helpful discussions. Part of this work was done while the author was visiting the Simons Institute for the Theory of Computing.

\appendix

\section{\texorpdfstring{\qnc}{QNC} simulation of \texorpdfstring{\qacf}{QACf} circuits} \label{s:qnc}

\begin{lem} \label{lem:qsim}
	For all $n$-qubit, depth-$d$ \qacf circuits $U$, there exists an $O(n)$-qubit, depth-$O(d \log n)$, size-$O(dn)$ \qnc circuit $C$ such that $C(I_n \otimes \ket\zs) = U \otimes \ket\zs$.
\end{lem}

\begin{proof}
	Green et al.~\cite{Gre+02} observed that the transformation $\ket{b, 0^{n-1}} \mapsto \ket{b^n}$ for $b \in \bits$ can be implemented by a size-$(n-1)$, depth-$\ceil{\log n}$ circuit consisting of CNOT gates with no ancillae. Therefore the fanout transformation $\ket{b, x} \mapsto \ket{b, x \oplus b^n}$ for $b \in \bits, x \in \cube n$ can be implemented by first computing $b^n$ as described above, then XORing $b^n$ onto $x$, and finally uncomputing $b^n$.
	
	Similarly an $n$-qubit generalized Toffoli gate can be cleanly simulated by a size-$O(n)$, depth-$O(\log n)$ \qnc circuit with $O(n)$ ancillae. This follows by simulating a log-depth DeMorgan formula for the AND function (i.e.\ the circuit whose graph is a balanced binary tree of 2-bit AND gates), with one ancilla qubit allocated to store the value of each gate in the DeMorgan formula, and then uncomputing the garbage.
	
	A general $n$-qubit, depth-1 \qacf circuit can be written as $\bigotimes_j G_j$, where each $G_j$ is a $k_j$-qubit gate such that $\sum_j k_j \le n$, and if $k_j>1$ then $G_j$ is either a generalized Toffoli or fanout gate. It follows that $\bigotimes_j G_j$ can be cleanly simulated by a \qnc circuit where the size and number of ancillae are $O\Paren{\sum_j k_j} \le O(n)$ and the depth is $O\Paren{\max_j \log k_j} \le O\Paren{\log \sum_j k_j} \le O(\log n)$. The lemma follows by successively implementing each layer of a \qacf circuit in this way, reusing the same ancillae to simulate each layer.
\end{proof}

\bibliographystyle{quantum}
\bibliography{Query_and_Depth_Upper_Bounds_for_Quantum_Unitaries_via_Grover_Search}

\begin{thebibliography}{10}

\bibitem{Aar16}
Scott Aaronson.
\newblock ``The complexity of quantum states and transformations: from quantum
  money to black holes''.
\newblock \href{https://arxiv.org/abs/1607.05256}{arXiv:1607.05256}~(2016).

\bibitem{AK07}
Scott Aaronson and Greg Kuperberg.
\newblock ``Quantum versus classical proofs and advice''.
\newblock \href{https://dx.doi.org/10.4086/toc.2007.v003a007}{Theory Comput.
  {\bf 3}, 129--157}~(2007).
\newblock
  \href{http://arxiv.org/abs/quant-ph/0604056}{arXiv:quant-ph/0604056}.

\bibitem{Aar21}
Scott Aaronson.
\newblock ``Open problems related to quantum query complexity''.
\newblock \href{https://dx.doi.org/10.1145/3488559}{ACM Trans. Quantum Comput.
  {\bf 2}, 1--9}~(2021).
\newblock  \href{http://arxiv.org/abs/2109.06917}{arXiv:2109.06917}.

\bibitem{LMW23}
Alex Lombardi, Fermi Ma, and John Wright.
\newblock ``A one-query lower bound for unitary synthesis and breaking quantum
  cryptography''.
\newblock In STOC.
\newblock \href{https://dx.doi.org/10.1145/3618260.3649650}{Pages 979--990}.
\newblock ~(2024).
\newblock  \href{http://arxiv.org/abs/2310.08870}{arXiv:2310.08870}.

\bibitem{Ros24}
Gregory Rosenthal.
\newblock ``Efficient quantum state synthesis with one query''.
\newblock In SODA.
\newblock \href{https://dx.doi.org/10.1137/1.9781611977912.89}{Pages
  2508--2534}.
\newblock ~(2024).
\newblock  \href{http://arxiv.org/abs/2306.01723}{arXiv:2306.01723}.

\bibitem{NC10}
Michael~A. Nielsen and Isaac~L. Chuang.
\newblock ``Quantum computation and quantum information: 10th anniversary
  edition''.
\newblock \href{https://dx.doi.org/10.1017/CBO9780511976667}{Cambridge
  University Press}. ~(2010).

\bibitem{DN06}
Christopher~M. Dawson and Michael~A. Nielsen.
\newblock ``The {S}olovay--{K}itaev algorithm''.
\newblock \href{https://dx.doi.org/10.26421/QIC6.1-6}{Quantum Inf. Comput. {\bf
  6}, 81--95}~(2006).
\newblock
  \href{http://arxiv.org/abs/quant-ph/0505030}{arXiv:quant-ph/0505030}.

\bibitem{INN+22}
Sandy Irani, Anand Natarajan, Chinmay Nirkhe, Sujit Rao, and Henry Yuen.
\newblock ``Quantum search-to-decision reductions and the state synthesis
  problem''.
\newblock In CCC.
\newblock \href{https://dx.doi.org/10.4230/lipics.ccc.2022.5}{Volume 234, pages
  5:1--5:19}.
\newblock ~(2022).
\newblock  \href{http://arxiv.org/abs/2111.02999}{arXiv:2111.02999}.

\bibitem{Sun+21}
Xiaoming Sun, Guojing Tian, Shuai Yang, Pei Yuan, and Shengyu Zhang.
\newblock ``Asymptotically optimal circuit depth for quantum state preparation
  and general unitary synthesis''.
\newblock \href{https://dx.doi.org/10.1109/TCAD.2023.3244885}{IEEE Trans.
  Comput.-Aided Des. Integr. Circuits Syst. {\bf 42}, 3301--3314}~(2023).
\newblock  \href{http://arxiv.org/abs/2108.06150}{arXiv:2108.06150}.

\bibitem{YZ23}
Pei Yuan and Shengyu Zhang.
\newblock ``Optimal (controlled) quantum state preparation and improved unitary
  synthesis by quantum circuits with any number of ancillary qubits''.
\newblock \href{https://dx.doi.org/10.22331/q-2023-03-20-956}{Quantum {\bf 7},
  956}~(2023).
\newblock  \href{http://arxiv.org/abs/2202.11302}{arXiv:2202.11302}.

\bibitem{ZLY22}
Xiao-Ming Zhang, Tongyang Li, and Xiao Yuan.
\newblock ``Quantum state preparation with optimal circuit depth:
  Implementations and applications''.
\newblock \href{https://dx.doi.org/10.1103/PhysRevLett.129.230504}{Physical
  Review Letters {\bf 129}, 230504}~(2022).
\newblock  \href{http://arxiv.org/abs/2201.11495}{arXiv:2201.11495}.

\bibitem{Gre+02}
Frederic Green, Steven Homer, Cristopher Moore, and Christopher Pollett.
\newblock ``Counting, fanout, and the complexity of quantum {ACC}''.
\newblock \href{https://dx.doi.org/10.26421/QIC2.1-3}{Quantum Inf. Comput. {\bf
  2}, 35--65}~(2002).
\newblock
  \href{http://arxiv.org/abs/quant-ph/0106017}{arXiv:quant-ph/0106017}.

\bibitem{HS05}
Peter H{\o}yer and Robert \v{S}palek.
\newblock ``Quantum fan-out is powerful''.
\newblock \href{https://dx.doi.org/10.4086/toc.2005.v001a005}{Theory Comput.
  {\bf 1}, 81--103}~(2005).

\bibitem{TT16}
Yasuhiro Takahashi and Seiichiro Tani.
\newblock ``Collapse of the hierarchy of constant-depth exact quantum
  circuits''.
\newblock \href{https://dx.doi.org/10.1007/s00037-016-0140-0}{Comput.
  Complexity {\bf 25}, 849--881}~(2016).
\newblock  \href{http://arxiv.org/abs/1112.6063}{arXiv:1112.6063}.

\bibitem{Has86}
Johan H{\aa}stad.
\newblock ``Almost optimal lower bounds for small depth circuits''.
\newblock In STOC.
\newblock \href{https://dx.doi.org/10.1145/12130.12132}{Pages 6--20}.
\newblock ~(1986).

\bibitem{Juk12}
Stasys Jukna.
\newblock ``Boolean function complexity''.
\newblock \href{https://dx.doi.org/10.1007/978-3-642-24508-4}{Volume~27 of
  Algorithms and Combinatorics}.
\newblock Springer, Heidelberg. ~(2012).

\bibitem{Lup58}
Oleg Lupanov.
\newblock ``On a method of circuit synthesis''.
\newblock \href{https://dx.doi.org/10.2307/2271493}{Izvestia VUZ {\bf 1},
  120--140}~(1958).

\bibitem{Sha49}
Claude Shannon.
\newblock ``The synthesis of two-terminal switching circuits''.
\newblock \href{https://dx.doi.org/10.1002/j.1538-7305.1949.tb03624.x}{Bell
  System Tech. J. {\bf 28}, 59--98}~(1949).

\bibitem{Amb02}
Andris Ambainis.
\newblock ``Quantum lower bounds by quantum arguments''.
\newblock \href{https://dx.doi.org/10.1006/jcss.2002.1826}{J. Comput. System
  Sci. {\bf 64}, 750--767}~(2002).
\newblock
  \href{http://arxiv.org/abs/quant-ph/0002066}{arXiv:quant-ph/0002066}.

\bibitem{Nay11}
Ashwin Nayak.
\newblock ``Inverting a permutation is as hard as unordered search''.
\newblock \href{https://dx.doi.org/10.4086/toc.2011.v007a002}{Theory Comput.
  {\bf 7}, 19--25}~(2011).
\newblock  \href{http://arxiv.org/abs/1007.2899}{arXiv:1007.2899}.

\bibitem{IB05}
S\'andor Imre and Ferenc Bal\'azs.
\newblock ``Quantum computing and communications: an engineering approach''.
\newblock \href{https://dx.doi.org/10.1002/9780470869048}{Chapter~7}.
\newblock John Wiley \& Sons. ~(2005).

\bibitem{Wie21a}
Nathan Wiebe~(2021).
\newblock Personal communication.

\end{thebibliography}

\end{document}